\documentclass[journal]{IEEEtran}
\usepackage{amsmath, amsfonts, amssymb, bm, amsthm}
\usepackage[utf8]{inputenc}
\usepackage{enumitem} 
\usepackage[ruled,vlined]{algorithm2e} 
\makeatletter
\renewcommand\@makefnmark{\hbox{\@textsuperscript{\normalfont\color{black}\@thefnmark}}}
\makeatother

\usepackage[bottom]{footmisc}
\usepackage{color, xcolor, soul, framed}
\usepackage{cancel}
\usepackage[numbers,sort&compress]{natbib}
\usepackage[font=normalfont]{caption} 

\usepackage{subcaption}
\usepackage{booktabs}
\usepackage{mathtools}
\usepackage{graphicx}
\usepackage{hyperref}
\usepackage{varwidth}

\newtheorem{prop}{Proposition}
\newtheorem{Definition}{Definition}
\newtheorem{Remark}{Remark}
\usepackage{multirow}
\usepackage{array}
\usepackage{lipsum}

\definecolor{turquoise}{rgb}{.0,.3,1.0}

\usepackage[first=0,last=9]{lcg}

\usepackage{color, colortbl}
\definecolor{Gray}{gray}{0.93}

\hypersetup{
    colorlinks=true,    citecolor=blue,    linkcolor=red,    filecolor=magenta,  urlcolor=blue,    pdftitle={Overleaf Example},
    }

\begin{document}
\title{Inertial-Based LQG Control: A New Look at Inverted Pendulum Stabilization}
\author{Daniel~Engelsman,~\IEEEmembership{Graduate Student Member,~IEEE} and~Itzik~Klein,~\IEEEmembership{Senior Member,~IEEE}
\thanks{The authors are with the Hatter Department of Marine Technologies, Charney School of Marine Sciences, University of Haifa, Israel. \{dengelsm@campus, kitzik@univ\}.haifa.ac.il}}
\maketitle

\begin{abstract}
Linear quadratic Gaussian (LQG) control is a well-established method for optimal control through state estimation, particularly in stabilizing an inverted pendulum on a cart. In standard laboratory setups, sensor redundancy enables direct measurement of configuration variables using displacement sensors and rotary encoders. However, in outdoor environments, dynamically stable mobile platforms—such as Segways, hoverboards, and bipedal robots—often have limited sensor availability, restricting state estimation primarily to attitude stabilization. Since the tilt angle cannot be directly measured, it is typically estimated through sensor fusion, increasing reliance on inertial sensors and necessitating a lightweight, self-contained perception module. Prior research has not incorporated accelerometer data into the LQG framework for stabilizing pendulum-like systems, as jerk states are not explicitly modeled in the Newton-Euler formalism. In this paper, we address this gap by leveraging local differential flatness to incorporate higher-order dynamics into the system model. This refinement enhances state estimation, enabling a more robust LQG controller that predicts accelerations for dynamically stable mobile platforms.
\end{abstract}
%
\begin{IEEEkeywords}
Linear quadratic Gaussian; optimal control; state estimation; differential flatness; inverted pendulum; 
\end{IEEEkeywords}
%

\section{Introduction}
\IEEEPARstart{S}oon after Rudolf E. Kálmán introduced the linear quadratic regulator (LQR) in 1960 \cite{kalman1960contributions}, the control community began exploring its application to dynamic systems. Of the many systems examined, the inverted pendulum on a cart (IPoC), commonly credited to Richard Bellman \cite{bertsekas2012dynamic}, stood out as particularly compelling. 
In this fundamental problem, a rigid pendulum is hinged at its base, allowing it to rotate with one degree of freedom (DoF), while a mobile cart, capable of horizontal movement, provides the translational DoF. By applying appropriate forces to the cart, the resulting accelerations counteract the pendulum’s deviations from the upright position, effectively stabilizing it and restoring equilibrium.
\\
This challenge spans several key disciplines within dynamical systems theory: i) Dynamics, where the motion of the pendulum and cart is governed by coupled nonlinear relationships but offers equilibrium points with local linearity; ii) Kinematics, which describes the system's motion over time, subject to positional setpoints and physical limitations; iii) Control, which requires a stabilizing strategy to counteract the pendulum's instability and prevent large deviations; and iv) Estimation, which focuses on determining the system's internal states, especially when full observability is lacking or when not all states can be directly measured \cite{prasad2014optimal}.
\\
As an unstable single-input, multiple-output (SIMO) system, its stabilization depends critically on the accuracy of the system model, the validity of the linearization around equilibrium points, the rapid minimization of both control and estimation errors, and, most importantly, robustness to noise, disturbances, and unmodeled uncertainties.
Over the years, various control strategies have been proposed, ranging from traditional proportional–integral–derivative (PID) and cascaded controllers \cite{nasir2008performance}, to more advanced techniques such as dynamic surface control \cite{huang2015nonlinear}, neural networks \cite{anderson2002learning}, information-aiding \cite{engelsman2023information}, and, more recently, reinforcement learning \cite{ozalp2020review}.
\\
Despite the importance of these approaches as control benchmarks, the LQG framework has remained dominant due to its analytical transparency, which ensures optimality when the system model is accurate and noise is assumed to follow a Gaussian distribution \cite{chacko2023lqr}. In practice, the true system states and their estimates are integrated within a state-space framework, facilitating the analysis of the closed-loop eigenvalues of the combined system, which includes both the Kalman filter (KF) and the controller's gain \cite{brunton2022data}.
\\
Since dynamically stable platforms are central to our study, the perception module depends on inertial sensor measurements, a mobile and standalone solution. However, this reliance also introduces instrumental noise that may compromise performance. To mitigate this issue, several approaches have been proposed, including bandpass filtering \cite{park2008error, alam2014adaptive}, learning-based frameworks \cite{li2011neural, engelsman2023data, cohen2024inertial}, multiple sensors \cite{waegli2010noise, libero2024augmented, engelsman2024parametric}, and, more recently, hybrid approaches \cite{chen2024slip, cohen2025adaptive}.
\\
Although much attention has been devoted to this problem, the challenge of integrating inertial sensor data directly into the LQG framework has remained unresolved outside controlled laboratory settings. To address this gap, our work makes the following contributions:
\begin{enumerate}[label=(\roman*)]
    \item \textbf{Inertial-aided LQG}: Extending the LQG framework with an acceleration-augmented process model to improve disturbance rejection and state estimation accuracy.
    \item \textbf{Higher-order dynamics}: Exploiting differential flatness to model linear and angular accelerations in terms of control inputs, improving system predictability.
    \item \textbf{Open-source}: Ensuring reproducibility by providing an interactive codebase implementation @ \href{https://github.com/ansfl/LQG-A-IPOC}{\texttt{\textbf{GitHub}}}.
\end{enumerate}

\noindent As will be shown later, the improved estimability leads to 27\%-39\% larger stability regions and a 10\%-15\% reduction in crash rates. The remaining sections of the paper are structured as follows: Section \ref{sec:theory} presents the theoretical background, Section \ref{sec:method} outlines our methodology, Section \ref{sec:results} provides and analyzes the results, and Section \ref{sec:conc} concludes the study.

\section{Theoretical Background} \label{sec:theory}
A dynamical system describes the evolution of a system's state over time. In general, a continuous-time dynamical system can be represented by the following differential equation
\begin{align}
\dot{\boldsymbol{x}}(t) = \boldsymbol{f}(\boldsymbol{x}(t), \boldsymbol{u}(t)) \ , \label{eq:sys}
\end{align}
where \( \boldsymbol{x}(t) \in \mathbb{R}^n \) represents the state vector, \( \boldsymbol{u}(t) \in \mathbb{R}^k \) denotes the control input vector, and $\boldsymbol{f}: \mathbb{R}^n \times \mathbb{R}^k \to \mathbb{R}^n$ is the state transition function, which defines the rate of change of the system’s state over time. The solution to this set of differential equations describes the trajectory of $\boldsymbol{x}(t)$ over time, starting from initial conditions $\boldsymbol{x}(0)$ and driven by occasional inputs $\boldsymbol{u}(t)$. The system is said to reach an equilibrium point, denoted by the subscript \( e \), when its dynamics no longer change and the states remain constant, defined by
\begin{align}
\dot{\boldsymbol{x}}(t) = \boldsymbol{f}(\boldsymbol{x}_e, \boldsymbol{u}_e) = \boldsymbol{0} \ . \label{eq:sys_eq}
\end{align} 
The nature of this equilibrium is analyzed through the concept of local linearity, which approximates the nonlinear dynamics \( \boldsymbol{f} \) by a linear model around of the equilibrium point \( (\boldsymbol{x}_e, \boldsymbol{u}_e) \). Based on the system's response to small perturbations, the equilibrium can be classified as: i) stable—if trajectories converge back to equilibrium, ii) unstable—if they diverge, or iii) saddle-point—if some perturbations lead the system back to equilibrium, while others cause divergence.
\\
In the analysis of the IPoC setup, stability is achieved when the configuration variables—namely, the pendulum angle and the cart position—remain constant, either at the bottom (unstable equilibrium) or the upright (stable equilibrium) positions. These steady-state conditions facilitate the description of the nonlinear system \eqref{eq:sys} around the equilibrium point \((\boldsymbol{x}_e, \boldsymbol{u}_e)\) as a linear approximation, via a first-order Taylor expansion \cite{arrowsmith1990introduction}, resulting in:
\begin{align}
\dot{\boldsymbol{x}}(t) \approx \mathbf{A}(\boldsymbol{x}_e, \boldsymbol{u}_e) (\boldsymbol{x}(t) - \boldsymbol{x}_e) + \mathbf{B}(\boldsymbol{x}_e, \boldsymbol{u}_e) (\boldsymbol{u}(t) - \boldsymbol{u}_e) , \label{eq:linearization}
\end{align}
where $\boldsymbol{x}(t)$ represent the state vector and $\boldsymbol{u}(t)$ the control inputs, and deviations from the equilibrium are given by
\begin{align}
\delta \boldsymbol{x}(t) = \boldsymbol{x}(t) - \boldsymbol{x}_e \ \ \text{and} \ \ \delta \boldsymbol{u}(t) = \boldsymbol{u}(t) - \boldsymbol{u}_e \ .
\end{align}
Assuming linear time-invariant (LTI) dynamics near the equilibrium points, the Jacobian matrices of the system with respect to the state and input are
\begin{align}
\mathbf{A}(\boldsymbol{x}_e, \boldsymbol{u}_e) = \frac{\partial \boldsymbol{f}}{\partial \boldsymbol{x}} \bigg|_{(\boldsymbol{x}_e, \boldsymbol{u}_e)} 
\, \text{and} \ \,
\mathbf{B}(\boldsymbol{x}_e, \boldsymbol{u}_e) = \frac{\partial \boldsymbol{f}}{\partial \boldsymbol{u}} \bigg|_{(\boldsymbol{x}_e, \boldsymbol{u}_e)} . \label{eq:Jacobians}
\end{align}
Since $\boldsymbol{x}_e$ and $\boldsymbol{u}_e$ are constants that vanish upon differentiation, the following identities hold
\begin{align}
\boldsymbol{x}(t) = \delta \boldsymbol{x}(t) + \boldsymbol{x}_e \ \ \text{and} \ \ \boldsymbol{u}(t) = \delta \boldsymbol{u}(t) + \boldsymbol{u}_e \ ,
\end{align}
simplifying the error dynamics to the standard LTI form
\begin{align}
\dot{\boldsymbol{x}}(t) = \mathbf{A} \, \boldsymbol{x}(t) + \mathbf{B} \, \boldsymbol{u}(t) \ . \label{eq:sys_lin}
\end{align}
To account for model discrepancies, the following zero-mean white Gaussian noise terms are sampled from their respective process (\(\mathbf{W}\)) and measurement (\(\mathbf{V}\)) noise covariance matrices
\begin{align}
\boldsymbol{w}(t) &\sim \mathcal{N}(\mathbf{0}, \mathbf{W}) \ , \quad \mathbb{E}[\boldsymbol{ww}^{T}] = \mathbf{W} \succ 0 \ , 
\\
\boldsymbol{v}(t) &\sim \mathcal{N}(\mathbf{0}, \mathbf{V}) \ \ , \quad \mathbb{E}[\boldsymbol{vv}^{T}] = \mathbf{V} \succeq 0 \ . 
\end{align}
This leads to the following stochastic state-space form
\begin{equation}
\begin{aligned} 
\dot{\boldsymbol{x}}(t) & = \mathbf{A} \, \boldsymbol{x}(t) + \mathbf{B} \, \boldsymbol{u}(t) + \boldsymbol{w}(t) \ ,  \\
\boldsymbol{y}(t) & = \mathbf{C} \, \boldsymbol{x}(t) + \boldsymbol{v}(t) \ , \label{eq:hidden}
\end{aligned}
\end{equation}
where the system states $\boldsymbol{x}(t) \in \mathbb{R}^n$ evolve according to the linear dynamics defined by matrices $\mathbf{A}$ and $\mathbf{B}$, and the observable outputs $\boldsymbol{y}(t) \in \mathbb{R}^m$ are a linear projection of the states through matrix $\mathbf{C}$, where typically $m<n$.

\subsection{Observer-based controller} 
Since only a subset of states are observable, a linear-quadratic estimator (LQE) is used to filter noise and estimate the unmeasured states. This allows the LQE to provide the feedback controller with real-time state estimates, as opposed to pure static feedback, where all states are assumed to be known.
In the context of LTI systems, their combination can be optimized by designing a linear-quadratic regulator (LQR) that minimizes penalties on both state deviations and control effort \cite{aastrom2012introduction, chrif2014aircraft}. 
\\
For simplicity, we omit the time argument, and the quadratic cost function over the time horizon $T$ is given by
\begin{align} \label{eq:cost}
\min_{\boldsymbol{u}(t)} J \triangleq \lim_{T \rightarrow \infty} \mathbb{E} \bigg[ \ \frac{1}{T} \int\limits_{0}^{T} \left( \boldsymbol{x}^\top \mathbf{Q} \boldsymbol{x} + \boldsymbol{u}^\top \mathbf{R} \boldsymbol{u} \right) dt \ \bigg] \ , 
\end{align}
subject to the state equations \eqref{eq:hidden}, where $\mathbf{Q}$ and $\mathbf{R}$ are positive semi-definite matrices that specify the relative weighting given to state deviations $\boldsymbol{x}(t)$ and control effort $\boldsymbol{u}(t)$, respectively. 
\\
Assuming that the pairs \( (\mathbf{A}, \mathbf{B}) \) and \( (\mathbf{A}, \mathbf{W}^{1/2}) \) are controllable, and \( (\mathbf{C}, \mathbf{A}) \) and \( (\mathbf{Q}^{1/2}, \mathbf{A}) \), are observable, the optimal feedback law that solves the linear quadratic Gaussian (LQG) problem is obtained by combining the LQR with a Kalman filter, which serves as the LQE. Such dynamics are given by
\begin{align}
\dot{\hat{\boldsymbol{x}}}(t) & = \mathbf{A} \, \hat{\boldsymbol{x}}(t) + \mathbf{B} \, \boldsymbol{u}(t) + \mathbf{L} \left( \boldsymbol{y}(t) - \mathbf{C} \hat{\boldsymbol{x}}(t) \right) \ , \label{eq:LQE} \\
\boldsymbol{u}(t) & = - \mathbf{K} \, \hat{\boldsymbol{x}}(t) \ , \label{eq:LQR}
\end{align}
where $\hat{\boldsymbol{x}}(t)$ is the estimated state. The optimal control gain matrix $\mathbf{K}$ is given by
\begin{align}
\mathbf{K} = \mathbf{R}^{-1} \mathbf{C}^\top \mathbf{S} \ ,
\end{align}
where $\mathbf{S}$ is the solution to the steady-state algebraic Riccati equation (ARE) in its matrix form
\begin{align}
\mathbf{A}^\top \mathbf{S} + \mathbf{S A} - \mathbf{S B R}^{-1} \mathbf{B^\top S + Q} = \boldsymbol{0} \ .
\end{align}
In a similar vein, the optimal KF gain $\mathbf{L}$ is given by
\begin{align}
\mathbf{L = P C^\top W}^{-1} \ ,
\end{align}
with $\mathbf{P}$ solving the associated ARE for the KF by
\begin{align}
\mathbf{A^\top P + P A - P C W}^{-1} \mathbf{C^\top P + V} = \boldsymbol{0} \ .
\end{align}
Let \( \boldsymbol{0} \) and \( \mathbf{I} \) represent the zero and identity matrices of the appropriate dimension, the true states \eqref{eq:hidden} and their corresponding estimates \eqref{eq:LQE} can be augmented as follows
\begin{align}
\frac{d}{dt}
\begin{bmatrix}
\boldsymbol{x} \\ \boldsymbol{\hat{x}}
\end{bmatrix} &= \begin{bmatrix}
    \mathbf{A} & \mathbf{-B K} \\ \mathbf{LC} & \mathbf{A-BK-LC}
\end{bmatrix} \begin{bmatrix} \boldsymbol{x} \\ \hat{\boldsymbol{x}} \end{bmatrix} + \begin{bmatrix} \mathbf{I} & \boldsymbol{0} \\ \boldsymbol{0} & \mathbf{L} \end{bmatrix} \begin{bmatrix} \boldsymbol{w} \\ \boldsymbol{v} \end{bmatrix} \, \notag , \\ 
\begin{bmatrix}
\boldsymbol{y} \\ \boldsymbol{u}
\end{bmatrix} &= \begin{bmatrix}
\mathbf{C} & \boldsymbol{0} \\ \boldsymbol{0} & -\mathbf{K}
\end{bmatrix} \begin{bmatrix} \boldsymbol{x} \\ \hat{\boldsymbol{x}} \end{bmatrix} + \begin{bmatrix} \boldsymbol{v} \\ \boldsymbol{0} \end{bmatrix} \ . \label{eq:states_est}
\end{align}
Fig.~\ref{fig:LQG} depicts a block diagram of the closed-loop system from \eqref{eq:states_est}, where the plant represents the linearized dynamics around the equilibrium points. The stability of the system is typically analyzed using the estimation error, defined as
\begin{align}
\boldsymbol{e}(t) = \boldsymbol{x}(t) - \hat{\boldsymbol{x}}(t) \ ,
\end{align}
which leads to the following block-state representation
\begin{align}
\begin{bmatrix}
    \dot{\boldsymbol{x}} \\ \dot{\boldsymbol{e}}
\end{bmatrix} = \begin{bmatrix}
    \mathbf{A-B K} & \mathbf{-B K} \\ \boldsymbol{0} & \mathbf{A-LC}
\end{bmatrix} \begin{bmatrix} \boldsymbol{x} \\ \boldsymbol{e} \end{bmatrix} + \begin{bmatrix} \mathbf{I} & \boldsymbol{0} \\ \mathbf{I} & -\mathbf{L} \end{bmatrix} \begin{bmatrix} \boldsymbol{w} \\ \boldsymbol{v} \end{bmatrix} \ . \label{eq:error}
\end{align}
The upper triangular structure of \eqref{eq:error} highlights the separation principle, as the observer dynamics ($\mathbf{A - LC}$) are decoupled from the closed-loop controller dynamics ($\mathbf{A - BK}$). Such isolation allows for independent stability analysis, ensuring that neither component interferes with the other while preserving overall closed-loop stability \cite{khalil1996robust, mahmudov2000controllability}. 
\begin{figure}[t] 
\begin{center}
\includegraphics[width=0.5\textwidth]{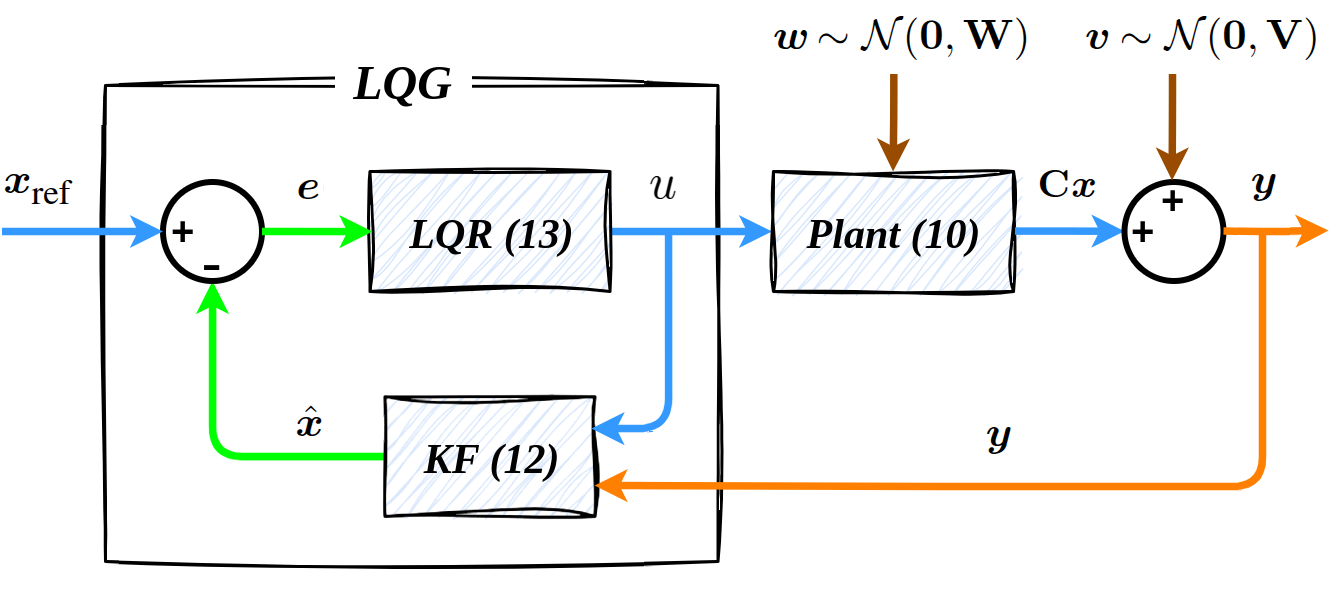}
\caption{LQG closed-loop dynamics diagram with a color-coded representation: true state (blue), measured state (orange), estimated state (green), and stochastic disturbances (brown).}
\label{fig:LQG}
\end{center}
\end{figure}

\subsection{Mathematical model}
The widely used nonlinear state equation from \cite{wang2010design, eide2011lqg, razmjooy2014comparison} describes the IPoC dynamics as
\begin{align}
    \begin{pmatrix}
        \dot{x} \\ \ddot{x} \\ \dot{\theta} \\ \ddot{\theta}
    \end{pmatrix} =
    \begin{pmatrix}
        \frac{d x}{dt} \\ 
        \frac{ -\frac{1}{2} m\text{g} \sin 2\theta + m \ell \, \dot{\theta}^2 \sin \theta -\delta \dot{x} + u}{(M+m(1- \cos^2 \theta))} \\
        \frac{d \theta}{dt} \\
        \frac{(M+m)\text{g} \sin\theta - \frac{1}{2} m \ell \, \dot{\theta}^2 \sin 2\theta + \cos \theta ( \delta \dot{x} + u)  }{(M+m(1- \cos^2 \theta) ) \ell }
    \end{pmatrix} . \label{eq:sys_1}
\end{align}
Here, \( x \) is the cart's horizontal position (mass \( M \)), \( \theta \) is the pendulum's angular position (mass \( m \), length \( \ell \)), \( \delta \) is the friction coefficient between them, and the scalar \( u \) is the acceleration command applied to the cart, as shown in Fig.~\ref{fig:pendulum}.
\\
Let $h$ represent the nonlinear measurement model that maps the system states in \eqref{eq:sys_1} to the output as follows
\begin{align}
\boldsymbol{y}(t) = \boldsymbol{h}(\boldsymbol{x}(t)) \ .
\end{align}
Under ideal settings, both configuration variables are measurable, and the linearized measurement model is expressed as
\begin{align}
\mathbf{C} = \frac{\partial \boldsymbol{h}}{\partial \boldsymbol{x}}\bigg|_{\boldsymbol{x}=\hat{\boldsymbol{x}}} = 
\begin{bmatrix}
1 & 0 & 0 & 0 \\
0 & 0 & 1 & 0 
\end{bmatrix} \ .
\end{align}
However, in real-world mobile pendulum-like systems—such as Segways, hoverboards, or biped robots—the configuration variables cannot be measured directly. Instead, they can typically be estimated from inertial sensors, albeit at the cost of a rank-deficient observability matrix.
\\
Traditional control strategies, such as full-state feedback and PID controllers, can directly integrate noisy acceleration measurements, prioritizing practicality over optimality. In contrast, the LQG framework relies on an explicit dynamical model for the effective data fusion. This approach enables the KF to compare state predictions with the actual measurements, thereby enhancing estimation accuracy.
\\
Building on this, the following section extends the current configuration described in \(\eqref{eq:sys_1}\) to incorporate higher-order derivatives into the state-space model, paving the way toward an inertial-aided LQG.
\begin{figure}[t] 
\begin{center}
\includegraphics[width=0.46\textwidth]{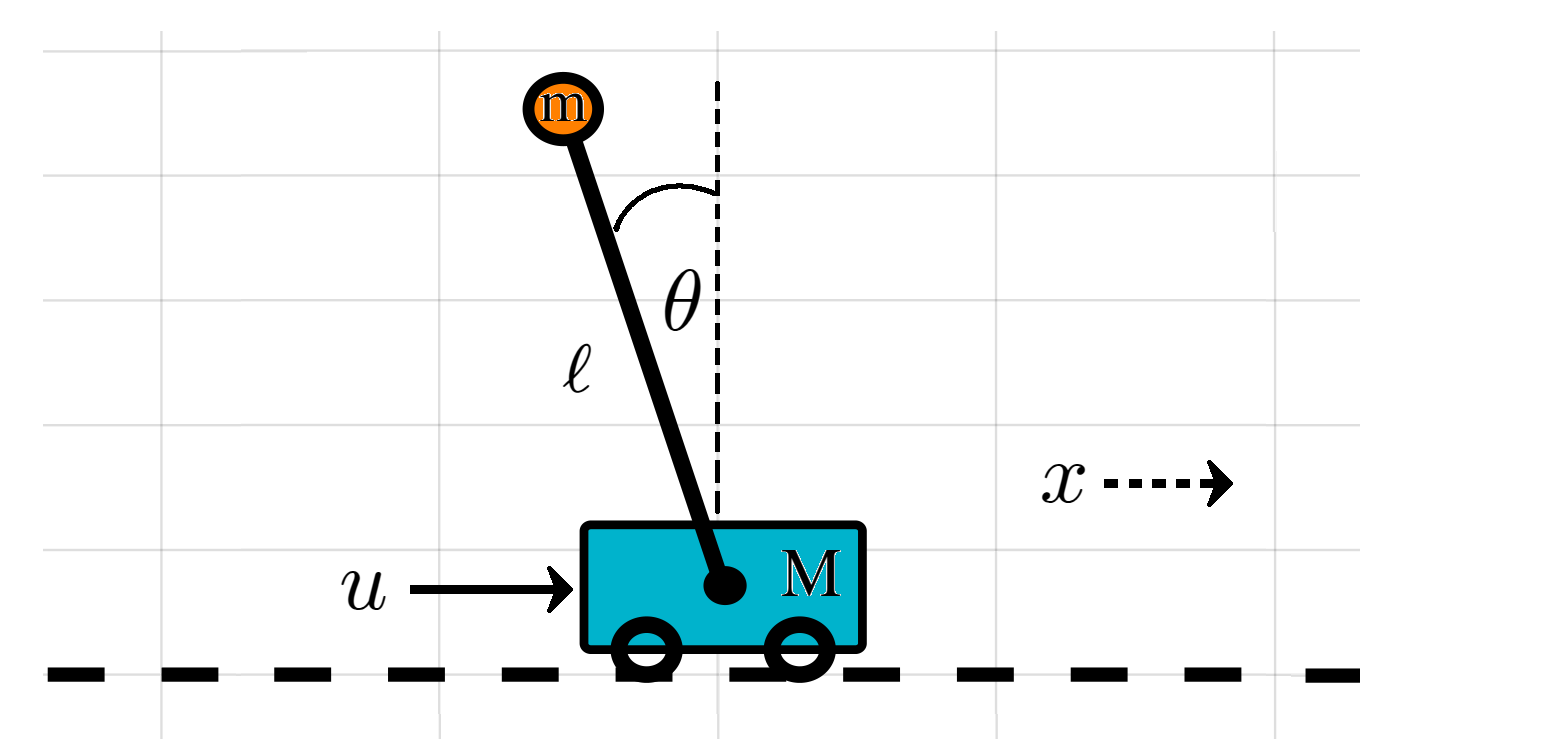}
\caption{A simplified mobile robot with an inverted pendulum system. Free body diagrams are commonly cited in \cite{wang2010design, eide2011lqg, razmjooy2014comparison, brunton2016koopman}.}
\label{fig:pendulum}
\end{center}
\end{figure}

\section{Proposed Methodology} \label{sec:method}
This section establishes the mathematical foundation for analyzing higher-order dynamics, which will serve as the basis for evaluating the LQG performance in the subsequent sections.

\subsection{Differential flatness} 
Differential flatness simplifies system dynamics and control by enabling direct representation in the flat output space, where states and inputs are explicitly determined from flat outputs and their derivatives, extending the notion of linear controllability to nonlinear systems.

\begin{Definition} \label{def_1}
A nonlinear system \eqref{eq:sys} with state vector \( \boldsymbol{x} \in \mathbb{R}^n \) and control input \( \boldsymbol{u} \in \mathbb{R}^m \), where \( m \leq n \), is differentially flat if there exists a set of flat outputs \( \boldsymbol\varepsilon \in \mathbb{R}^m \) such that
\begin{align}
\boldsymbol\varepsilon = \boldsymbol{h} ( \boldsymbol{x}, \boldsymbol{u}, \dot{\boldsymbol{u}}, \dots, \boldsymbol{u}^{(p)}) \ , \label{eq:phi_y}
\end{align}
allowing the state and control input to be expressed without requiring integration as
\begin{align}
\boldsymbol{x} = \boldsymbol\varphi_x(\boldsymbol\varepsilon, \dot{\boldsymbol\varepsilon}, \ddot{\boldsymbol\varepsilon}, \dots, \boldsymbol\varepsilon^{(q)}) \ , \label{eq:phi_x} 
\\
\boldsymbol{u} = \boldsymbol\varphi_u(\boldsymbol\varepsilon, \dot{\boldsymbol\varepsilon}, \ddot{\boldsymbol\varepsilon}, \dots, \boldsymbol\varepsilon^{(q)}) \ , \label{eq:phi_u}
\end{align}
where $p$ and $q$ are finite integers, and \( \boldsymbol\varphi \) is any smooth, \(q\)-times differentiable mapping function  \cite{fliess1995flatness, murray1995differential, van1998differential, rigatos2015nonlinear}. 
\end{Definition}

\begin{prop} \label{Prop:I}
The controlled IPoC system defined in \eqref{eq:sys_1}, \( \dot{\boldsymbol{x}}(t) = \boldsymbol{f}(\boldsymbol{x}(t), u(t)) \), demonstrates differential flatness.
\end{prop}
\begin{proof}
As a single-input system (\( u \in \mathbb{R} \)), the flatness property permits the selection of only one flat output (\( \varepsilon \in \mathbb{R} \)). 
Unlike higher-order derivatives, the configuration variables are differentially independent, making it challenging to establish a direct mapping between them. 
\\
As shown in Fig.~\ref{fig:pendulum}, this challenge can be addressed using a geometric approach that captures both the pendulum and the cart through their horizontal displacement. This quantity is chosen as the flat output and is defined by
\begin{align}
\varepsilon = x + \ell \sin \theta \ .\label{eq:df_0}
\end{align}
With two unknowns, an additional independent equation is required to uniquely solve for both state variables. To achieve this, the horizontal component of gravity is used to express
\begin{align}
\ddot{\varepsilon} = \text{g} \sin \theta \ . \label{eq:df_1}
\end{align}
In this manner, when in the upright position (\(\theta = 0\)), the displacement acceleration is zero, whereas at full deflection, it reaches the full gravitational acceleration, as shown below
\begin{align}
\ddot{\varepsilon}({\theta = 0}) &= 0 \ , \\
\ddot{\varepsilon}({\theta = \pm {\pi}/{2}}) &= \text{g} \ .
\end{align}
Next, the pendulum angle is extracted from \eqref{eq:df_1}, and substituting it into \eqref{eq:df_0} yields the cart position, resulting in both
\begin{align}
\theta &= \arcsin({\frac{\ddot{\varepsilon}}{\text{g}}}) \ , \\
x &= \varepsilon - \frac{\ddot{\varepsilon} \ell}{\text{g}} \ .
\end{align}
Consequently, their derivatives are readily obtained as follows
\begin{align}
\dot{x} &= \dot{\varepsilon} - \frac{\dddot{\varepsilon} \ell}{\text{g}} \ , \label{eq:df_2} \\
\dot{\theta} &= \frac{\dddot{\varepsilon}}{\sqrt{\text{g}^2 - \ddot{\varepsilon}^2}} \ , \label{eq:df_3}
\end{align}
thus satisfying \eqref{eq:phi_x}. Next, based on \eqref{eq:sys_1}, the control input $u$ appears to be governed by
\begin{align}
\ddot{x} &= \frac{ -\frac{1}{2} m\text{g} \sin 2\theta + m \ell \, \dot{\theta}^2 \sin \theta -\delta \dot{x} + u}{(M+m(1- \cos^2 \theta))} \ , \\
\ddot{\theta} &= \frac{(M+m)\text{g} \sin\theta - \frac{1}{2} m \ell \, \dot{\theta}^2 \sin 2\theta + \cos \theta ( \delta \dot{x} + u)  }{(M+m(1- \cos^2 \theta) ) \ell } \, .
\end{align}
%
With three unknowns across two equations, we differentiate \eqref{eq:df_2} and \eqref{eq:df_3} to eliminate the state derivatives by
\begin{align}
\ddot{x} &= \frac{d}{dt} \dot{x}(\varepsilon) = \ddot{\varepsilon} - \frac{\varepsilon^{\textit{(4)}} \ell}{\text{g}} \ , \label{eq:df_6} \\ 
\ddot{\theta} &= \frac{d}{dt} \dot{\theta}(\varepsilon) = \frac{\varepsilon^{\textit{(4)}}({\text{g}^2 - \ddot{\varepsilon}^2)} + \dddot{\varepsilon}^2\ddot{\varepsilon}}{({\text{g}^2 - \ddot{\varepsilon}^2)^{3/2} }} \ . \label{eq:df_7}
\end{align}
%
Substituting \eqref{eq:df_6} and \eqref{eq:df_7} allows expressing the control input solely in terms of the flat outputs, thereby satisfying \eqref{eq:phi_u} as
\begin{align}
u = \varphi_u( \varepsilon, \dot{\varepsilon}, \ddot{\varepsilon}, \dddot{\varepsilon}, {\varepsilon}^{\textit{(4)}} ) \ .
\end{align}
Lastly, postulate \eqref{eq:phi_y} dictates that the mapping must be invertible in both directions, ensuring that the flat outputs are solely functions of the states and control inputs.
\\
To achieve this, all higher-order derivatives of the flat output must be determined, beginning with
\begin{align}
\dot{\varepsilon} = \frac{d}{dt} \varepsilon(t) = \dot{x} + \dot{\theta} \ell \cos \theta \ , 
\end{align}
followed by differentiation of \eqref{eq:df_6} which gives
\begin{align}
\dddot{x} &= \frac{d}{dt} \ddot{x}({\varepsilon}) = \dddot{\varepsilon} - \frac{{\varepsilon}^{\textit{(5)}}}{\text{g}} \ ,
\end{align}
and similarly, differentiating \eqref{eq:df_7} yields
\begin{align}
\dddot{\theta} = \frac{d}{dt} \ddot{\theta}(\varepsilon) = \frac{ \left( \alpha \, \varepsilon^{\textit{(5)}} + \dddot{\varepsilon}^3 \right) \alpha +
3 \, \beta ( \alpha \, \varepsilon^{\textit{(4)}} + \beta \dddot{\varepsilon} ) } {\alpha^{5/2}} \, ,
\end{align}
where $\alpha = ( \text{g}^2 - \ddot{\varepsilon}^2 )$, $\beta = (\ddot{\varepsilon} \dddot{\varepsilon})$, and $\gamma = M + m (1 - \cos^2 \theta)$ are defined for brevity.
\\
To remove the functional dependence on the third derivatives, their time expressions must be derived as well, while linear jerk is obtained by
\begin{align} \label{eq:jerk_lin}
\dddot{x} &= \left( - m\text{g} \dot{\theta} \cos 2\theta + m \ell \, \dot{\theta} (2 \ddot{\theta}\sin \theta + \dot{\theta}^2 \cos \theta ) - \delta \ddot{x} + \dot{u} \right) / {\gamma} \notag \\
+ & \left( \frac{1}{2} m\text{g} \sin 2\theta - m \ell \, \dot{\theta}^2 \sin \theta + \delta
\dot{x} - u \right) \frac{m \dot{\theta} \sin 2\theta }{\gamma^2} \, , 
\end{align}
and similarly, the angular jerk is derived by
\begin{align} \label{eq:jerk_ang}
\dddot{\theta} &= \bigg( (M+m)\text{g} \dot{\theta}\cos \theta - m\ell (\dot{\theta}\ddot{\theta} \sin 2\theta + \dot{\theta}^3 \cos 2\theta) \notag \\
- & \dot{\theta} \sin \theta (\delta \dot{x} + u) + \cos \theta (\delta \ddot{x} + \dot{u})  \bigg) / ({\gamma \ell}) - \frac{m \dot{\theta} \sin 2\theta}{\gamma^2 \ell} \notag \\
 \bigg( &(M+m)\text{g} \sin \theta -\frac{1}{2} m \ell \, \dot{\theta}^2 \sin 2 \theta + \cos \theta (\delta \dot{x} + u) \bigg) \, ,
\end{align}
where $\gamma(\theta)=(M + m \sin^2 \theta)$. 
\\
With an equal number of equations and unknowns, all state derivatives can be substituted, satisfying \eqref{eq:phi_y} as
\begin{align}
\varepsilon = \varphi_x( x, u, \dot{u} ) \ .
\end{align}
This confirms that the mapping defined by $\varphi_x$ and $\varphi_u$ is locally invertible, such that $\varepsilon$ and its derivatives uniquely determine the system states and inputs
\begin{align}
\begin{pmatrix}
\varepsilon \\ {\vdots} \\ \varepsilon^{\textit{(5)}}
\end{pmatrix}  \ \Leftrightarrow \ 
\begin{pmatrix}
x \\ u \\ \dot{u}
\end{pmatrix} \, ,
\end{align}
thereby establishing a diffeomorphism between the state-input space and the flat-output space, which ensures differential flatness as defined in Definition~\ref{def_1}.
\end{proof}
%
Following Proposition~\eqref{Prop:I}, the solutions of the jerk states introduced in \eqref{eq:jerk_lin} and \eqref{eq:jerk_ang} are incorporated into the augmented IPoC (A-IPoC) state vector
\begin{align}
\boldsymbol{x} 
 = \big[ \ x \quad \dot{x} \quad \ddot{x} \quad \theta \quad \dot{\theta} \quad \ddot{\theta} \ \big]^\top \, , 
\end{align}
which evolves according to the following dynamics
\begin{align}
\dot{ \boldsymbol{x} } = 
\begin{pmatrix}
    \frac{d x}{dt} \\ 
    \frac{ -\frac{1}{2} m\text{g} \sin 2\theta + m \ell \, \dot{\theta}^2 \sin \theta -\delta \dot{x} + u}{(M+m \sin^2 \theta)} \\
    \frac{d \ddot{x} }{dt} \\
    \frac{d \theta}{dt} \\
    \frac{(M+m)\text{g} \sin\theta - \frac{1}{2} m \ell \, \dot{\theta} \sin 2\theta + \cos \theta ( \delta \dot{x} + u)  }{(M+m \sin^2 \theta) \ell } \\
    \frac{d \ddot{\theta} }{dt}
\end{pmatrix} . \label{eq:sys_jerk}
\end{align}

\subsection{Linearized Dynamics}
While the nonlinear motion model captures the system's global behavior, its linearization is crucial for providing local validity in control design.
Pendulum-like systems exhibit two equilibrium points, as described in \eqref{eq:sys_eq}, along their trajectory: i) at the bottom position, a stable equilibrium where small perturbations cause the system to return to this point, creating a basin of attraction; and ii) at the upright position, an unstable equilibrium where any deviation results in divergence unless active control is applied.
\\
According to \eqref{eq:sys_lin}, applying a first-order approximation to the A-IPoC model \eqref{eq:sys_jerk} yields the state Jacobian, given by
\begin{align}
\mathcal{A}&(\boldsymbol{x}_e, \boldsymbol{u}_e) = \frac{\partial \boldsymbol{f}}{\partial \boldsymbol{x}} \bigg|_{(\boldsymbol{x}_e, \boldsymbol{u}_e)} = \notag \\
& 
\begin{pmatrix}
0&        1&        0&                  0&                  0& 0 \\
0& -\frac{\delta}{M}&    0&    -\frac{m \text{g}}{M} &                  0& 0 \\
0&        0&     -\frac{\delta}{M}&                  0&      \frac{m \text{g}}{M}& 0 \\
0&        0&        0&                  0&                  1& 0 \\
0& -\frac{\delta}{M \ell}&        0& -\frac{(M + m)\text{g}}{M \ell}&                  0& 0 \\
0&        0& -\frac{\delta}{M \ell}&                  0& -\frac{(M + m)\text{g}}{M \ell} & 0
\end{pmatrix} \ , 
\end{align}
and similarly, the control Jacobian is given by
\begin{align}
\mathcal{B}(\boldsymbol{x}_e, \boldsymbol{u}_e) = \frac{\partial \boldsymbol{f}}{\partial \boldsymbol{u}} \bigg|_{(\boldsymbol{x}_e, \boldsymbol{u}_e)} = 
\begin{pmatrix}
0 & 0 \\
\frac{1}{M} & 0 \\
0 & \frac{1}{M} \\
0 & 0 \\
-\frac{1}{M \ell} & 0 \\
0 & -\frac{1}{M \ell}
\end{pmatrix}  \ . 
\end{align}
Hence, linearizing around the upright equilibrium point,
\begin{align}
(\boldsymbol{x}_e, \boldsymbol{u}_e) = \left(
\begin{bmatrix}
    x_{\text{ref}} & 0 & 0 & \pi & 0 & 0 
\end{bmatrix}^\top , \begin{bmatrix} 0 & 0 \end{bmatrix}^\top \right) \, ,
\end{align}
results in the linearized A-IPoC dynamics expressed as
\begin{align}
\dot{\boldsymbol{x}} = \mathcal{A} \, \boldsymbol{x} + \mathcal{B} \, \boldsymbol{u} \ . 
\end{align}
%
Consequently, the derived measurement model is given by 
\begin{align} \label{eq:sys_obs_adapt}
\mathcal{C} = \frac{\partial \boldsymbol{h}}{\partial \boldsymbol{x}}\bigg|_{\boldsymbol{x}=\hat{\boldsymbol{x}}} = 
\left(
\begin{array}{cccccc}
  1 & 0 & 0 & 0 & 0 & 0 \\
  \cmidrule[.3mm]{1-6}
  0 & 0 & 1 & 0 & 0 & 0 \\
  0 & 0 & 0 & 0 & 1 & 0 \\
\end{array}
\right) \ , 
\end{align}
emphasizing two operational modes: continuous inertial sensing at the bottom, facilitated by accelerometers and gyroscopes, and infrequent position data at the top, when available. This establishes the system's output as
\begin{align}
\boldsymbol{y} = \mathcal{C} \, \boldsymbol{x} \ . 
\end{align}
\begin{Remark} 
To derive the system's Jacobians at the pendulum's bottom equilibrium, negate the signs of the entries in the 5th and 6th rows of both the matrices $\mathcal{A}$ and $\mathcal{B}$.
\end{Remark} 
%

\subsection{Stability region} \label{sec:Monte}
Controlling the IPoC involves two simultaneous objectives: (i) preventing the pendulum from falling and (ii) steering the system toward the desired position, all while maintaining an upright posture. While stability guarantees have been extensively explored in \cite{bittanti1991simple, arelhi1997lqg, barya2010comparison}, these analytical methods assumed noise-free conditions and overlooked the KF performance and its coupling with the controller, as demonstrated in \eqref{eq:error}.
\\
To address this, Algorithm~\ref{alg:MC} proposes a statistical sampling procedure in which Monte Carlo simulations are used to vary the initial conditions, generating a 2D stability map. 
\\
Initially, linear and angular velocities are randomly sampled from the admissible sample space \( \boldsymbol{\Gamma}_0 \), defined as
\begin{align}
   \boldsymbol{\Gamma}_0 = \dot{\boldsymbol{x}}_0 \times \dot{\boldsymbol{\theta}}_0 = [-10 \, , 10] \times \left[-\pi \, , \pi \right] \ ,
\end{align}
with each identified stable pair being stored in ($\dot{\boldsymbol{x}}_s, \dot{\boldsymbol{\theta}}_s$). 
\\
After \( i_{\text{max}} \) iterations, outliers are discarded, and interpolation is performed between the outermost points to create a closed convex hull. 
The boundaries of this hull mark the LQG marginal stability, with the inner (stable) region $\mathcal{S}$ defined as
\begin{align}
\mathcal{S} = \text{conv} \left( \dot{\boldsymbol{x}}_s \, , \, \dot{\boldsymbol{\theta}}_s \right) \subseteq \boldsymbol\Gamma_0 \ .
\end{align}

\begin{algorithm}[h] 
\caption{Identifying stability boundaries through Monte Carlo random sampling.}
\SetAlgoLined
\For{$i \leftarrow 1$ to $i_{\text{max}}$}{
    $\dot{x}_0 \gets {U}(-10, 10)$; $\dot{\theta}_0 \gets {U}(-\pi, \pi)$; \\
    $\boldsymbol{x}_0 \gets \big( x_0, \dot{x}_0, {\theta}_0, \dot{\theta}_0 \big)$ \hspace{12mm} \tcp{Initialize} 
    \textsf{solve} \hspace{.8mm}
    $\dot{\boldsymbol{x}}(t) = f(\boldsymbol{x}(t), u(t))$ \\
    %
    %
    \eIf{$\boldsymbol{x}(t \gg 0) \, \rightarrow \boldsymbol{x}_e$}{
        $  \mathcal{S}_{[ \dot{x}_0, \dot{\theta}_0 ]} \gets \mathcal{S}_{[ \dot{x}_0, \dot{\theta}_0 ]} + 1$ \hspace{5mm} \tcp{Stable}
    }
    {
        $  \mathcal{S}_{[ \dot{x}_0, \dot{\theta}_0 ]} \gets \mathcal{S}_{[ \dot{x}_0, \dot{\theta}_0 ]} - 1$ \hspace{5mm} \tcp{Unstable} 
    }
} 
\KwRet{$\mathcal{S}(\mathcal{S}>0)$} \hspace{9mm} \tcp{Stability region} \label{alg:MC}
\end{algorithm}


\section{Results and Analysis} \label{sec:results}
This section begins with validating the LQG controller's performance using the proposed A-IPoC motion model, followed by a comparison with the conventional IPoC implementation.

\subsection{Time-domain performance metrics} \label{sec:metrics}
The effectiveness of the proposed LQG controller is assessed across two distinct time frames. First, the transient response is extracted from the system outputs, and key performance metrics—including peak time (\( t_p \)), transient time (\( t_{\text{tr}} \)), and settling time (\( t_s \))—are computed based on a \( \pm 2\% \) tolerance of the output's final value. The total control effort exerted over the operational period (\( T \)) is then computed as
\begin{align}
\text{U}_{tot} = \int_0^T \| u(t) \| \, dt \, . \label{eq:u_tot}
\end{align}
To account for mechanical limitations, the actuator reaches its peak output (saturation) at $u_{\text{max}} = 3 \text{g}$ [m/s$^2$], which is modeled by the following clipping function
\begin{align}
\text{sat}(u) =  
\begin{cases}  
\ \ u_{\text{max}}, & \ \, u \ > u_{\text{max}} \\  
\ \ u,             & \hspace{1mm} |u| \leq u_{\text{max}} \\  
-u_{\text{max}}, & \hspace{2mm} u \ < -u_{\text{max}} 
\end{cases} \ . 
\end{align}
%
Once the transients have decayed, the end-point performance is evaluated by analyzing the steady-state error ($e_{ss}$). Both transient and steady-state metrics are then combined to provide an overall performance assessment. The first metric considered is the integral of absolute error (IAE), defined as
\begin{align}
    \text{IAE} = \int_0^T |e(t)| \, dt \ ,
\end{align}
treating both large and small deviations equally as undesirable. 
The second metric, the Integral of time-weighted absolute error (ITEA), emphasizes long-term stability by integrating the absolute error weighted by time
\begin{align}
\text{ITEA} = \int_0^T t \cdot |e(t)| \, dt \ .
\end{align}
Finally, to ensure consistency and mitigate numerical effects, the normalized time-to-go expression is employed to facilitate the visualization of spatial performance
\begin{align}
\Bar{t}_{\text{go}} = \frac{T-t}{T} \quad \forall \quad t \leq T \ . \label{eq:t_go}
\end{align}
\begin{table}[h]
\centering
\caption{Physical parameters used in all simulations.}
\renewcommand{\arraystretch}{1.35}
\begin{tabular}{|c|c|c|c|}
\hline 
Symbol & Parameter & Value & Unit \\ \hline
m & Pendulum mass & 1.0 & kg \\
M & Cart mass & 5.0 & kg \\
$\text{g}$ & Local gravity & 9.81 & ${\text{m}}/{\text{s}^2}$ \\ 
$\ell$ & Pendulum length & 1.25 & m \\
$\delta$ & Friction coefficient & 0.8 & ${\text{kg}}/{\text{s}}$ \\
\hline \hline
$T$ & Total time & 15 & [s] \\
$\Delta t$ & Time interval & 0.005 & [s] \\
\hline
\end{tabular} \label{t:specs}
\end{table} 
\begin{figure}[b]
\begin{center}
\includegraphics[width=0.5\textwidth]{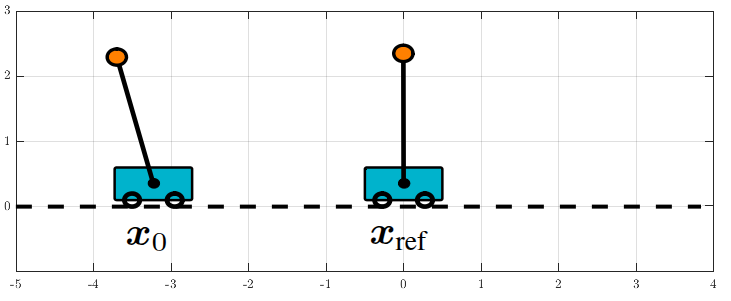}
\caption{Stabilization task: From initial state to target position.}
\label{f:pendulum}
\end{center}
\end{figure} 
\\
Table~\ref{t:specs} lists the simulation parameters used in this study, while Fig.~\ref{f:pendulum} depicts the stabilization task. The pendulum starts in a tilted and displaced state, $\boldsymbol{x}_{0}$, and the controller aims to guide it steadily to the reference target at the origin. Consequently, the setpoint vector is defined as $\boldsymbol{x}_{\text{ref}} = \mathbf{0} \in \mathbb{R}^6$, where any nonzero entries are interpreted as errors.

\subsection{Performance analysis}
We begin by evaluating the A-IPoC performance, with a key focus on the update ratio, denoted as \( \rho \), which determines the frequency of corrections following state predictions. 
\\
Figure~\ref{fig:states_1} illustrates the ideal baseline of \( \rho = 1 \), where each prediction step is immediately followed by an update, ensuring stable estimation of all six states in the model. Within the linearized region, the LQG performs well: the Kalman filter (KF) estimates remain closely aligned with the true states, while the LQR generates timely control commands that accurately steer the states toward the reference target.
\begin{figure}[t]
\begin{center}
\includegraphics[width=0.495\textwidth]{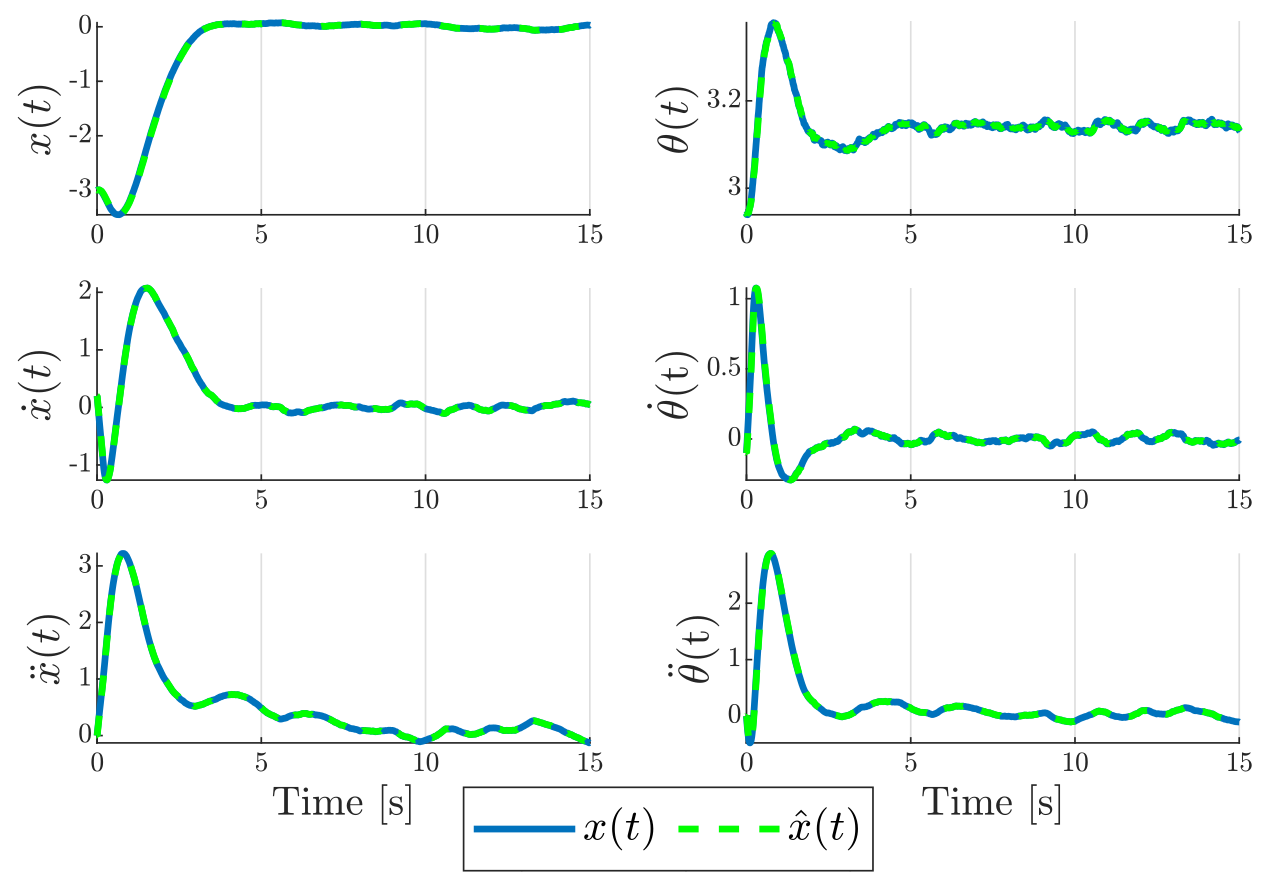}
\caption{LQG response for $\rho =1$. Left: Linear states; Right: Angular states. Blue: True states; Green: Estimated states.}
\label{fig:states_1}
\end{center}
\end{figure} 
\\
However, in real-world scenarios, communication errors result with measurement dropouts, effectively reducing the update ratio (\(\rho < 1\)). As a result, prolonged unaided periods lead to drift—commonly known as dead reckoning—forcing the controller to rely on increasingly inaccurate state estimates. 
\begin{figure}[b]
\begin{center}
\includegraphics[width=0.495\textwidth]{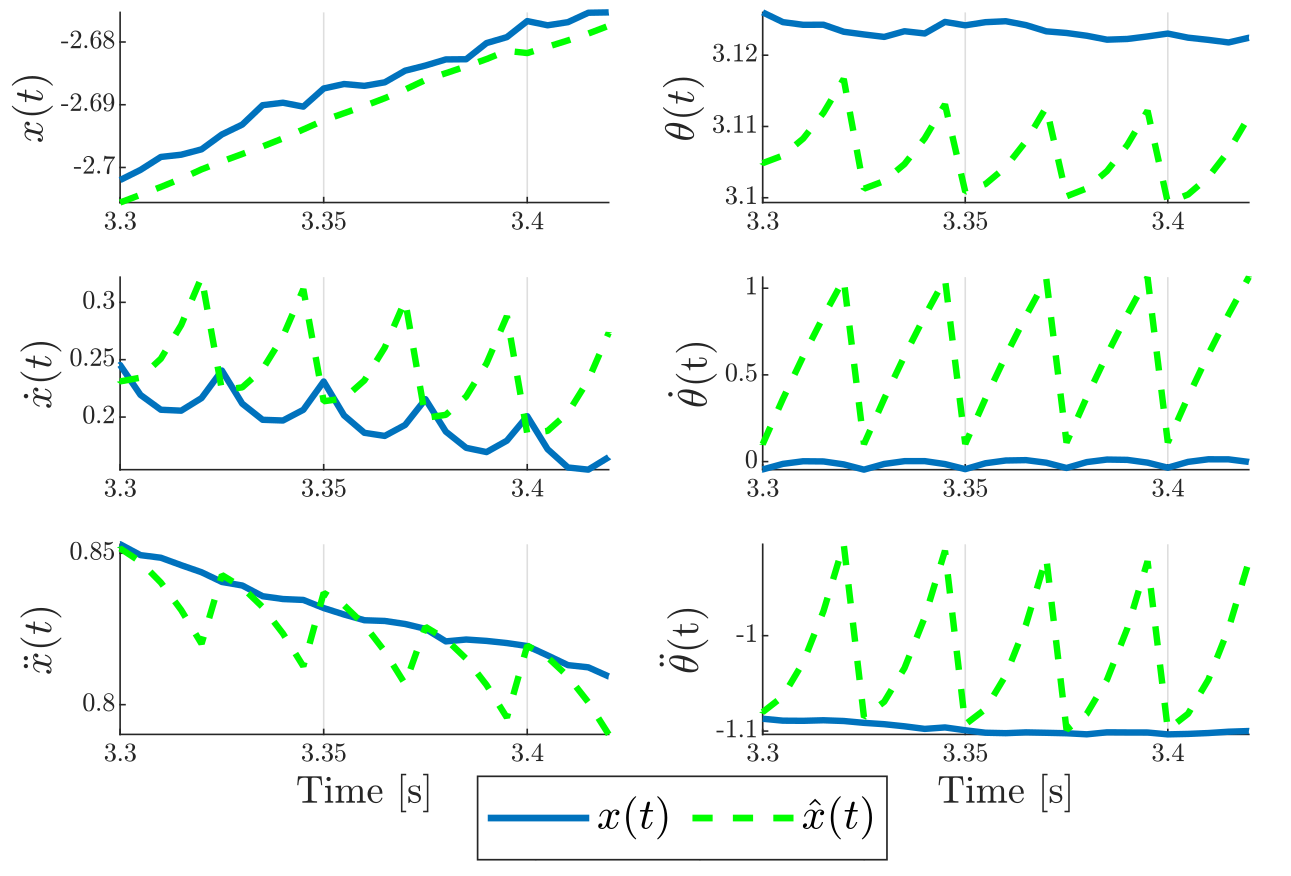}
\caption{LQG response for $\rho = 0.2$. Left: Linear states; Right: Angular states. Blue: True states; Green: Estimated states.}
\label{fig:states_2}
\end{center}
\end{figure} 
\\
Fig.~\ref{fig:states_2} illustrates this effect when the update ratio drops to 1:5 (\(\rho = 0.2\)), exposing the LQG to significantly fewer updates.
As observed, all subfigures display varying levels of two key effects: i) discrepancy—between the extrapolated estimates and true states, characterized by the sawtooth-shaped profile (dashed green line), and ii) instability—of the true states themselves (blue), resulting from the controller's underperformance due to a lack of external information.
\\
Fig.~\ref{fig:states_3} examines the sensitivity thresholds of the observer-controller (KF-LQR) coupling across different update ratios. In the upper row (\(\rho = 0.5\)), small and bounded errors are observed, with the corresponding control action displayed in the rightmost column. This performance results from the predictive model mitigating uncertainty, maintaining stability despite the reduced update frequency.
\\
In the middle row, with an increased update discontinuity (\(\rho = 0.1\)), the error patterns become thicker and noisier. This is reflected in the control action, which exhibits higher fluctuations as it struggles to maintain stability in the configuration variables. The result is a stable but unsmooth response, albeit at the cost of increased control effort.
\begin{figure}[t]
\begin{center}
\includegraphics[width=0.49\textwidth]{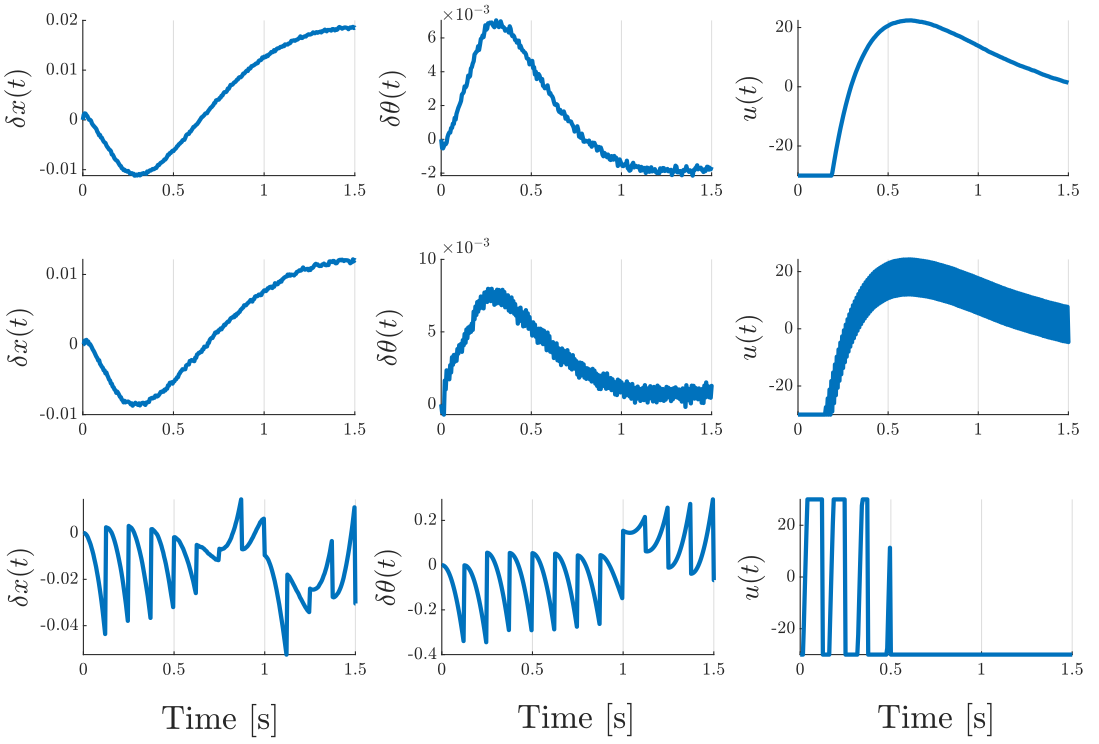}
\caption{Left: position error, Middle: Angular error, Right: Control signal. Update ratios (top to bottom): $\rho = \{ 0.5, 0.1, 0.01 \}$.}
\label{fig:states_3}
\end{center}
\end{figure} 
\begin{table}[b]
\centering
\caption{Transient response analysis of the LQG controller.}
\renewcommand{\arraystretch}{1.25}
\begin{tabular}{|c|c|c|c||c|c|}
\hline 
$\rho$ & $t_p$ [s] & $t_{tr}$ [s] & $t_s$ [s] & $u_{sat}$ [\%] & $\text{U}_{tot}$ [N$\cdot$s] \\ \hline 
\multirow{2}{*}{1.0} & 2.98 & 3.06 & 9.45 & \multirow{2}{*}{0} & \multirow{2}{*}{54.34} \\ 
& 0.98 \cellcolor{Gray} & 4.82 \cellcolor{Gray} & 5.63\cellcolor{Gray} & &  \\ \hline 
\multirow{2}{*}{0.5} & 5.14 & 7.81 & 10.75 & \multirow{2}{*}{0} & \multirow{2}{*}{98.13} \\ 
& 1.20 \cellcolor{Gray} & 6.64 \cellcolor{Gray} & 7.13 \cellcolor{Gray} & & \\ \hline 
\multirow{2}{*}{0.1} & 7.52 & 14.75 & - & \multirow{2}{*}{2.12} & \multirow{2}{*}{178.21} \\ 
& 3.31 \cellcolor{Gray} & 13.21 \cellcolor{Gray} & -\cellcolor{Gray} & & \\ \hline 
\multirow{2}{*}{0.05} & 12.01 & - & - & \multirow{2}{*}{13.31} & \multirow{2}{*}{320.96} \\ 
& 4.40 \cellcolor{Gray} & -\cellcolor{Gray} & -\cellcolor{Gray} & & \\ \hline 
\multirow{2}{*}{0.01} & - & - & - & \multirow{2}{*}{94.20} & \multirow{2}{*}{$>$ 1e4} \\ 
& -\cellcolor{Gray} & -\cellcolor{Gray} & -\cellcolor{Gray} & & \\ \hline 
\end{tabular} \label{t:err_1}
\end{table} 
\\
Finally, as the update ratio reaches its minimum value (\(\rho = 0.01\)), the controller exhibits abrupt switching between saturation boundaries, indicating a loss of stability. The state estimates drive the system toward divergence, and the increasing drift leads to inaccurate or insufficient control feedback, causing perturbations the controller cannot recover from. 
\\
Table~\ref{t:err_1} presents numerical data supporting these arguments, based on the metrics outlined in Sec.~\ref{sec:metrics}, including position errors (top, unfilled), angular errors (bottom, shaded in gray), and the percentage of actuator saturation (rightmost column).
\\
As expected, performance gradually degrades as \(\rho\) decreases, leading to longer durations to reach the \(\pm 2\)\% error band. 
\\
While the controller can stabilize the pendulum for mid-range \(\rho\) values, the cart's position drift expands its trajectory, resulting in higher energy expenditures. 
\\
At the extreme, as shown in the bottom rows, smaller \(\rho\) values induce pronounced oscillations that prevent the system from settling (blank cells) or, in some cases, lead to instability, as indicated by actuator saturation.
%
%
\\
Next, we investigate the balance between the two weighting matrices $\mathbf{Q}$ and $\mathbf{R}$, as defined in \eqref{eq:cost}. The interplay between these matrices is crucial for optimizing the trade-off between state error minimization and control effort, leading to the identification of four distinct weighting profiles \cite{brunton2016koopman}:
\begin{itemize}  
\item Poorly chosen $\mathbf{Q}$ and $\mathbf{R}$: Suboptimal performance, excessive control effort, and potential instability.  
\item \( \mathbf{Q} \gg \mathbf{R} \): Rapid and aggressive response with high control effort, potentially leading to actuator saturation.  
\item \( \mathbf{R} \gg \mathbf{Q} \): Energy-efficient control with reduced actuation effort but at the cost of sluggish response.  
\item Well-tuned $\mathbf{Q}$ and $\mathbf{R}$: Ideal effort-performance trade-off, promoting stability and efficient actuation.
\end{itemize}
While the last case is generally preferred, achieving this balance involves delicate tuning, which often conflicts with other system requirements. As a result, we adjust our weightings to reach a balanced compromise, referred to as 'ours'.
\\
Table~\ref{t:profiles} presents various tuning weights tailored to specific use cases. For simplicity, the matrices $\mathbf{Q}$ and $\mathbf{R}$ are defined using vectors of ones, \( \mathbf{1}_n \), each scaled uniquely to achieve an appropriate time-energy trade-off\footnote{See code @ {\small \tt{\url{https://GitHub.com/ANSFL/LQG-A-IPOC}}}}.
%
\begin{table}[h]
\centering
\caption{Heuristic tuning for different operational modes.}
\renewcommand{\arraystretch}{1.4}
\begin{tabular}{|c|c|c||c|c|}
\hline 
Mode & diag($\mathbf{Q}$) & diag($\mathbf{R}$) & Time [s] & $\text{U}_{tot}$ [N$\cdot$s] \\ \hline
Low-power & $0.1 \cdot \mathbf{1}_6$ & $10 \cdot \mathbf{1}_2$ & 21.23 & 0.51\\ \hline
Utility & $\mathbf{1}_6$ & $\mathbf{1}_2$ & 14.32 & 10.86 \\ \hline
\textbf{Ours} & $\mathbf{1}_6$ & $0.1 \cdot \mathbf{1}_2$ & 9.45 & 54.34 \\ \hline
Agile & $10 \cdot \mathbf{1}_6$ & $0.01 \cdot \mathbf{1}_2$ & 2.88 & 653.40 \\ \hline 
\end{tabular} \label{t:profiles}
\end{table}

\subsection{Performance comparison}
Having evaluating the performance of the LQG controller for our proposed A-IPoC, we will now assess whether it provides meaningful improvements over the conventional IPoC setup. Fig.~\ref{fig:conv_position} compares both LQG implementations, focusing on the time evolution of the configuration variables across various values of \(\rho\).
%
\begin{figure}[t]
\begin{center}
\includegraphics[width=0.5\textwidth]{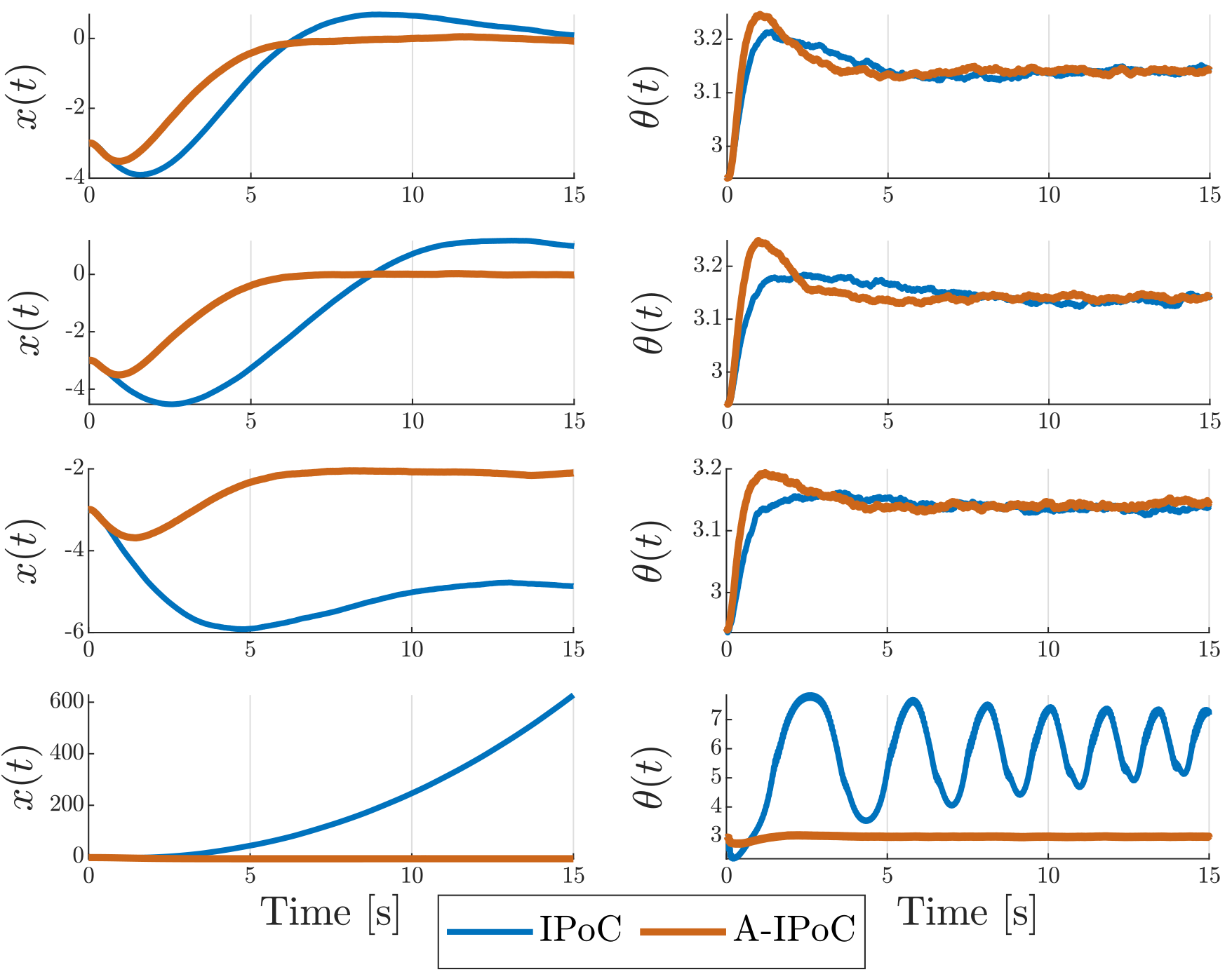}
\caption{Configuration errors: IPoC vs. A-IPoC. Update ratios (top to bottom): \( \rho = \{1.0, 0.5, 0.2, 0.01\} \).}
\label{fig:conv_position}
\end{center}
\end{figure} 
%
For continuous updates (\(\rho = 1\)), no significant difference is observed between the models. However, as \(\rho\) decreases, the IPoC exhibits a more sluggish response in the position variables, while the A-IPoC remains slightly more responsive, reaching the reference position more quickly. 
\newpage
These observations give rise to two key insights:
\begin{itemize}
    \item \textbf{Initial undershoot dynamics}: The non-minimum phase behavior causes the cart to briefly move in the opposite direction before eventually converging to the reference value, rather than directly reaching it. This counterintuitive motion complicates control, as it requires additional effort to counteract the initial deviation.
    \item \textbf{Different error bounds}: While the cart’s position error can grow unbounded, the pendulum’s error is inherently constrained by the cyclic nature of angles, limiting its divergence during filter instability. This is evident in the right-hand-side subfigures, where the LQG effectively stabilizes all angular variables but struggles with the cart's position, even at the same $\rho$ values.
    \item \textbf{Higher stability margin}: At the lowest \(\rho\) value, the conventional IPoC-based controller exhibits divergence in both configuration variables, whereas the A-IPoC-based LQG controller remains stable.
\end{itemize}
%
%
\begin{table}[b]
\centering
\caption{Performance metrics: IPoC vs. A-IPoC (ours).}
\renewcommand{\arraystretch}{1.35}
\begin{tabular}{|c|c|c|c|c|c|c|} 
\multirow{2}{*}{$\rho$} & \multicolumn{2}{c|}{IAE} & \multicolumn{2}{c|}{ITAE} & \multicolumn{2}{c|}{$|\,  e_{ss} \,  |$} \\ 
 & IPoC & A-IPoC & IPoC & A-IPoC & IPoC & A-IPoC \\ \hline 
\multirow{2}{*}{1.0} & 23.35 & \textbf{11.64} & 77.13 & \textbf{26.32} & 0.06 & \textbf{0.03} \\ 
& 0.31\cellcolor{Gray} & \textbf{0.24}\cellcolor{Gray} & 1.13\cellcolor{Gray} & \textbf{0.24}\cellcolor{Gray} & \textbf{0.002}\cellcolor{Gray} & 0.004\cellcolor{Gray} \\ \hline 
\multirow{2}{*}{0.5} & 110.8 & \textbf{50.58} & 881.7 & \textbf{370.2} & 7.25 & \textbf{3.28}\\ 
& 0.26\cellcolor{Gray} & \textbf{0.16}\cellcolor{Gray} & 0.84\cellcolor{Gray} & \textbf{0.43}\cellcolor{Gray} & 0.006\cellcolor{Gray} & 0.006\cellcolor{Gray} \\ \hline 
\multirow{2}{*}{0.2} & 160.4 & \textbf{69.57}& - & \textbf{531.7} & 11.86 & \textbf{4.71}\\ 
& 0.36\cellcolor{Gray} & \textbf{0.18}\cellcolor{Gray} & 1.27\cellcolor{Gray} & \textbf{0.50}\cellcolor{Gray} & \textbf{0.003}\cellcolor{Gray}& 0.004\cellcolor{Gray} \\ \hline 
\multirow{2}{*}{0.1} & - & \textbf{129.7}& - & \textbf{847.4}& - & \textbf{9.49}\\ 
& -\cellcolor{Gray} & \textbf{0.46}\cellcolor{Gray} & -\cellcolor{Gray} & \textbf{0.96}\cellcolor{Gray} & -\cellcolor{Gray}& \textbf{0.008}\cellcolor{Gray} \\ \hline 
\end{tabular} \label{t:err_2}
\end{table} 
To enhance our temporal analysis, both variables are normalized by their initial conditions \(\dot{x}_0, \dot{\theta}_0\), and then against the normalized time axis, as described in \eqref{eq:t_go}. This projection improves interpretability and ensures that all trajectories start from the light blue sphere at (1,1,1) and, if stabilization is achieved, converge to the red sphere at the origin (0,0,0); otherwise, they diverge.
\\
Fig.~\ref{fig:conv_angular} extends the previous analysis across similar sets of \(\rho\) values, providing a clearer comparison by highlighting both differences in transient behavior and variations in steady-state errors. While the observed patterns are consistent with earlier findings, the normalization effect more clearly emphasizes how error magnitudes increase as stability deteriorates. 
\\
Table~\ref{t:err_2} compares all simulation results, where the IAE and ITAE metrics assess cumulative error propagation, \(e_{ss}\) quantifies the residual error over the time span \(T\), and the better-performing model is highlighted in bold.  
\\
While both models exhibit comparable steady-state errors, the A-IPoC achieves significantly lower transient errors, reducing them by 30\%–80\% compared to the IPoC. As illustrated in both figures above, the A-IPoC maintains marginal stability even under extreme unaided conditions (\(\rho=0.01\)), where the IPoC completely diverges.  
These results confirm the A-IPoC's improved responsiveness, error correction, and tracking performance, highlighting its superior stability margin.
\begin{figure}[t]
\begin{center}
\includegraphics[width=0.495\textwidth]{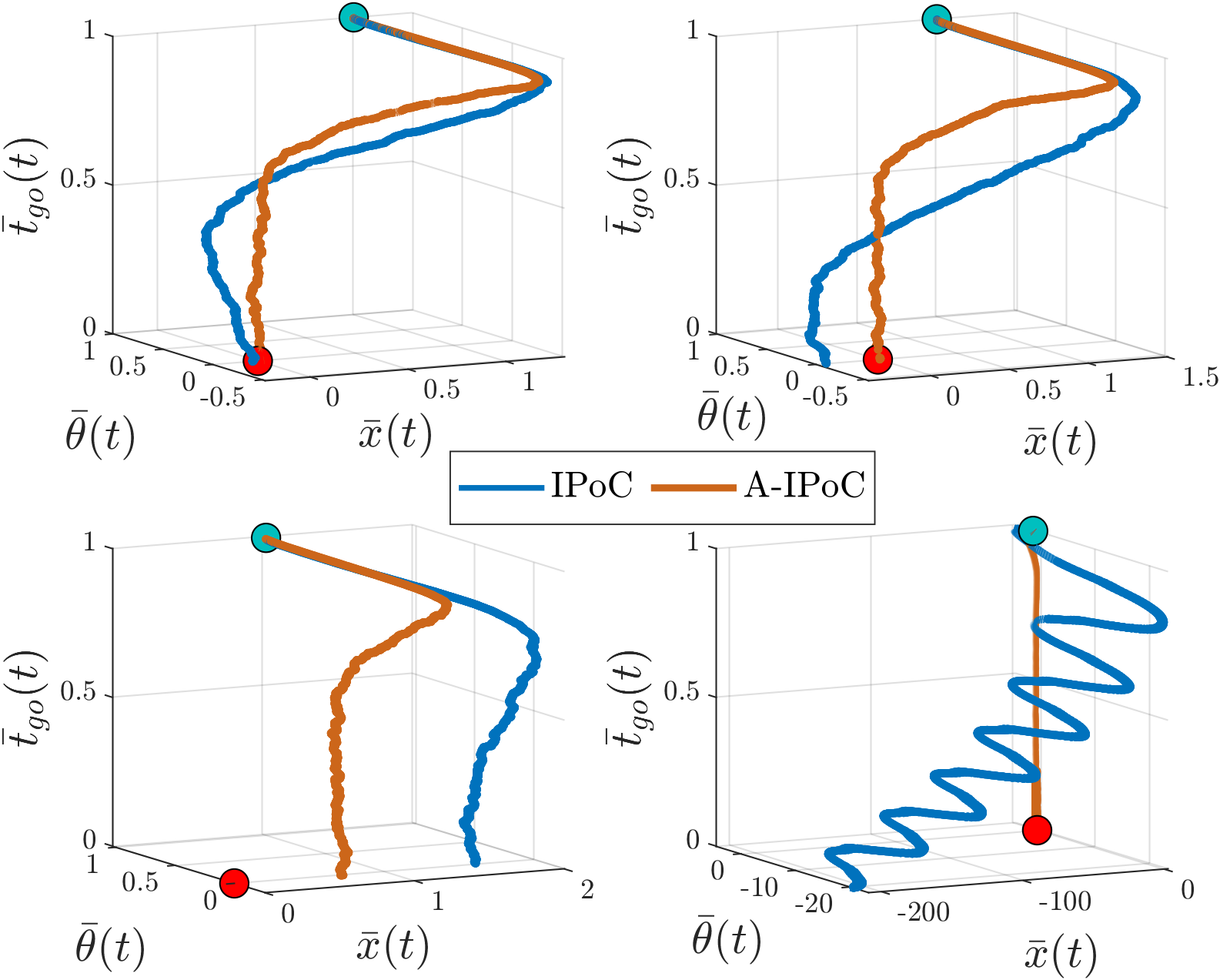}
\caption{Normalized trajectories: IPoC vs. A-IPoC. Update ratios (left to right, top to bottom): \( \rho = \{1.0, 0.5, 0.2, 0.01\} \).}
\label{fig:conv_angular}
\end{center}
\end{figure} 

\subsection{Stability comparison}
While analytical stability guarantees for continuous loop closure (\(\rho = 1\)) are well-documented in the literature, our work challenges this assumption by introducing update discontinuities (\(\rho < 1\)) and control saturation, which complicate closed-form analysis. Instead, we use Algorithm~\hyperref[sec:Monte]{1} to heuristically explore the stable regions where the LQG controller can successfully re-stabilize the system across a wide range of initial perturbations, \(\boldsymbol{\Gamma}_0 = \dot{\boldsymbol{x}}_0 \times \dot{\boldsymbol{\theta}}_0\) .
\begin{figure}[t]
\begin{center}
\includegraphics[width=0.5\textwidth]{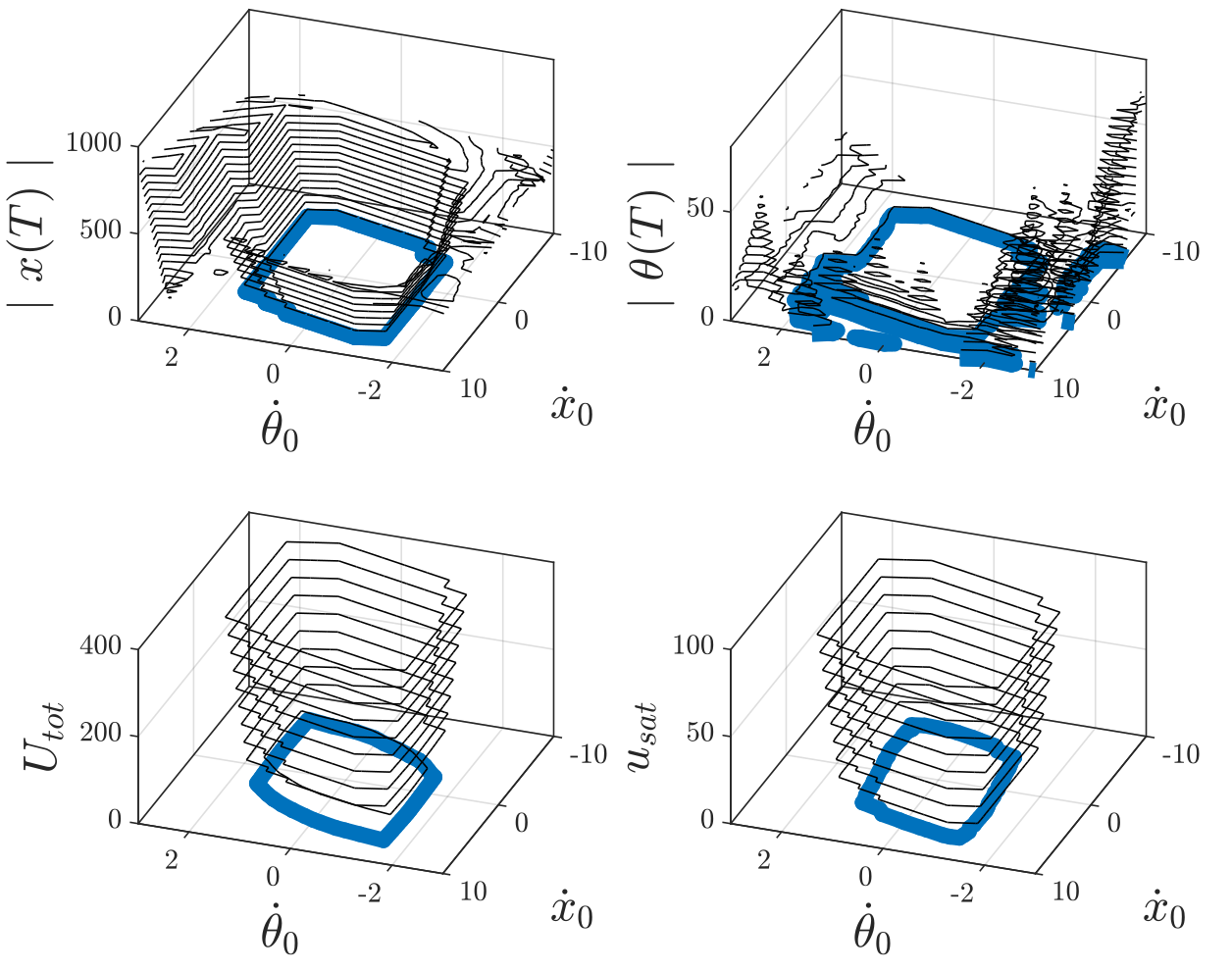}
\caption{Response surfaces of the LQG-based IPoC system.}
\label{fig:stab_1}
\end{center}
\end{figure}
\\
Fig.~\ref{fig:stab_1} presents four response surfaces in a clockwise sequence: i) final cart position (\(x(T)\)), ii) final pendulum angle (\(\theta(T)\)), iii) total control effort (\(\text{U}_{tot}\) based on \eqref{eq:u_tot}), and iv) actuator saturation percentage over time (\(u_{sat}\)). 
\\
Before comparing the models, each subfigure provides valuable insights into stability characteristics from different perspectives:
\begin{enumerate}
    \item \textbf{Basin of attraction}: The iterative exploration identifies plateau-like regions (highlighted by the thick blue line at the bottom of the surface), which define the boundedness of both state and control outputs. A larger polygon area for the \( j \) metric, denoted as \( \mathcal{S}_j \), indicates greater resistance to divergence.
    \item \textbf{Instability thresholds}: As the system surpasses these boundaries, the output level sets exhibit rapid growth, initially manifesting as amplified fluctuations before escalating into divergence. This visualization not only delineates the critical stability thresholds but also identifies directional components that characterize the system's sensitivity to perturbations within the broader state space.
    \item \textbf{Energy efficiency}: While the top row focuses on stability, the bottom rows show the associated energy expenditure. Although the stability regions (\(\mathcal{S}\)) retain a similar shape, both control costs and saturation times increase significantly as instability progresses. This is reflected in the widening of the topographic level sets, indicating the growing energy needed to counteract rising errors and restore stability.
\end{enumerate}
\begin{figure}[t]
\begin{center}
\includegraphics[width=0.495\textwidth]{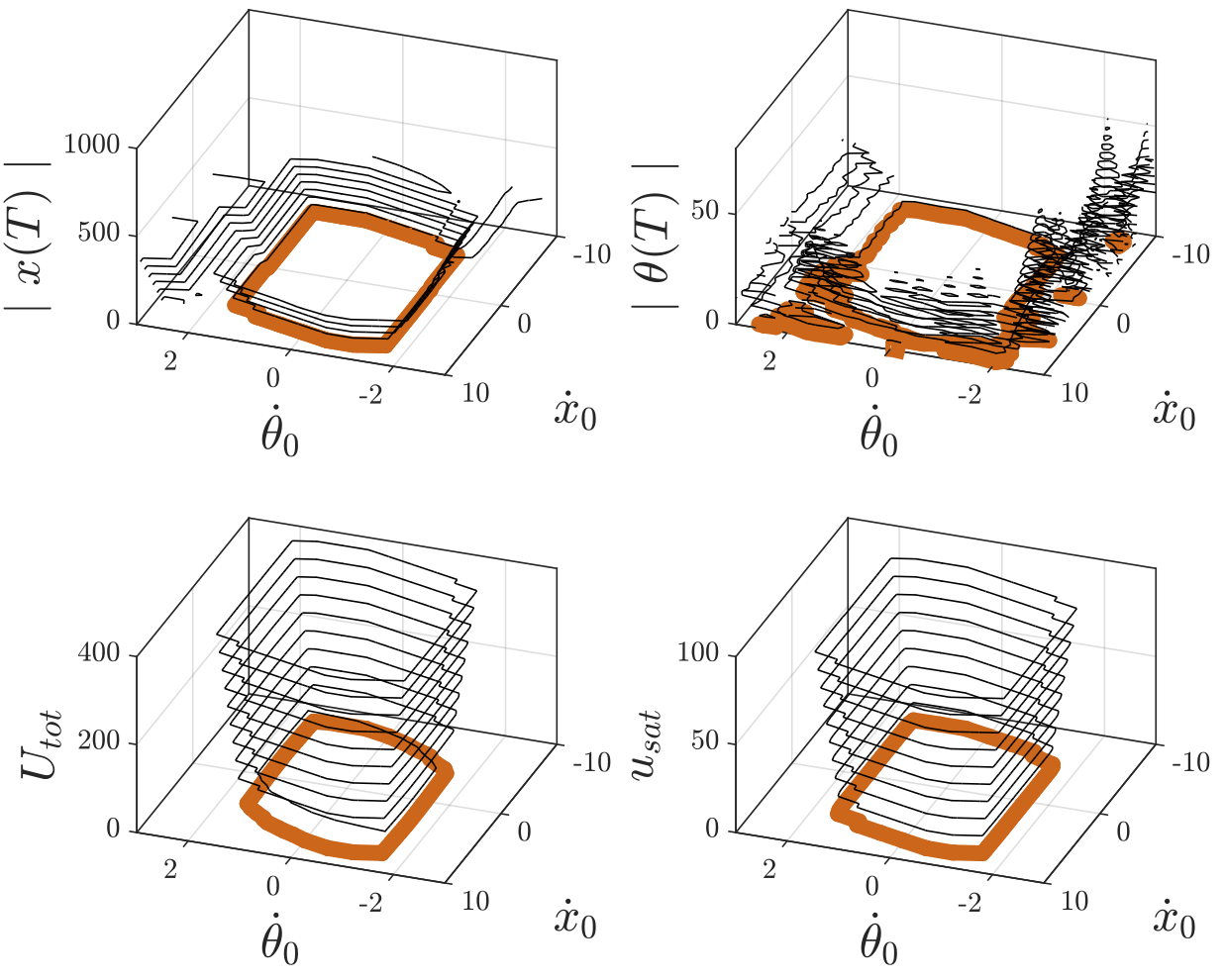}
\caption{Response surfaces of the LQG-based A-IPoC system.}
\label{fig:stab_2}
\end{center}
\end{figure} 
When compared directly with Fig.~\ref{fig:stab_2}, both models show similar patterns, with only subtle differences. However, a closer inspection reveals two key aspects that highlight significant contrasts.
\\
First, the stability region shows a notable difference, with the A-IPoC implementation consistently exhibiting larger basin (stability) areas compared to its IPoC counterpart (\(\mathcal{S}_{\text{A-IPoC}} > \mathcal{S}_{\text{IPoC}}\)). Second, the extent to which the level sets are expanded upward differs between the two models, with the IPoC divergence level being higher. 
\\
This suggests that the A-IPoC model offers greater marginal stability, as its divergence remains more moderate for the same initial perturbation values.
\\
Table~\ref{t:err_3} presents these results numerically, using normalization techniques to eliminate physical units. For each of the four criteria, denoted by \( j \), the corresponding basin area is represented by \( \mathcal{S}_j \). To aid in assessment, the first column shows the fraction \( \mathcal{S}_j \) relative to the basin area obtained from the IPoC model. The second column normalizes the \( j \)-th basin area by the total permissible perturbation range, \( \boldsymbol{\Gamma}_0 \), highlighting the relative robustness of each model. The rightmost column quantifies the total number of instances where stability is lost for both configuration variables, complementing the analysis in the previous columns. 
\begin{table}[h]
\centering
\caption{Key performance metrics: IPoC vs. A-IPoC (ours).}
\renewcommand{\arraystretch}{1.35}
\begin{tabular}{|c|c|c|c|c|c|c|} 
\multirow{2}{*}{Criteria} & \multicolumn{2}{c|}{$\mathcal{S}_j/\mathcal{S}_{\text{IPoC},j}$} & \multicolumn{2}{c|}{$\mathcal{S}_j \, / \, \boldsymbol{\Gamma}_0$} & \multicolumn{2}{c|}{Crash rate [\%]} \\ 
 & IPoC & A-IPoC & IPoC & A-IPoC & IPoC & A-IPoC \\ \hline 
$| \, x(T) \, |$ & 1 & \textbf{1.27} & 0.469 & \textbf{0.595} & 53.07 & \textbf{40.50} \\ \hline 
$| \, \theta(T) \, |$ & 1 & \textbf{1.28} & 0.539 & \textbf{0.688} & 46.15 & \textbf{31.24}\\ \hline 
$u_{\text{sat.}}$ & 1 & \textbf{1.31} & 0.402 & \textbf{0.525} & 59.80 & \textbf{47.50} \\ \hline 
$\text{U}_{tot}$ & 1 & \textbf{1.39} & 0.354 & \textbf{0.491} & 64.62 & \textbf{50.91}\\ \hline 
\end{tabular} \label{t:err_3}
\end{table} 
\\
Clearly, the A-IPoC outperforms the IPoC, exhibiting stability areas 27\%-39\% larger and crash rates 10\%-15\% lower. As a result, the A-IPoC sustains a broader stable region, providing superior immunity to disturbances and control limits under identical conditions.

\section{Conclusion} \label{sec:conc}
This study presents a significant enhancement to the LQG control framework for pendulum-like systems, addressing a long-standing challenge in integrating accelerometer measurements directly into predictive motion models within the classical Newton-Euler formulation. By leveraging differential flatness theory, we introduce an approach that augments the system’s dynamics with higher-order terms (A-IPoC), enabling more accurate acceleration predictions and smoother controller actuation.  
\\
Our method improves state extrapolation during sensor-limited periods, strengthening observer-controller coupling and leading to 27\%-39\% larger stability regions and 10\%-15\% lower crash rates. Across a range of dynamic and high-uncertainty verification scenarios, the proposed framework demonstrated superior robustness, exhibiting greater immunity to instability and enhanced disturbance rejection. These results bridge a crucial gap in existing research and establish a foundation for more reliable and efficient control of dynamically stable systems, particularly those with non-minimum phase dynamics, which are often hindered by inherent delays and inverted control behavior.


\section*{Acknowledgment}
D.E. is supported by the Maurice Hatter foundation and by the Bloom School Institutional Excellence Scholarship for outstanding doctoral students at the University of Haifa.

\bibliographystyle{IEEEtran}

\bibliography{Ref}

\end{document}